\documentclass{article}

\usepackage{amsmath}
\usepackage{amssymb}
\usepackage{amsfonts}
\usepackage{graphicx}
\usepackage{subfigure}
\usepackage{authblk}
\title{Phase transitions in filtration of real gases}
\author[1,2]{Valentin Lychagin}
\author[1,3]{Mikhail Roop}
\affil[1]{V.A. Trapeznikov Institute of Control Sciences, Russian Academy of Sciences, 65 Profsoyuznaya Str., 117997 Moscow, Russia}
\affil[2]{Department of Mathematics, The Arctic University of Norway, Postboks 6050, Langnes 9037, Tromso, Norway}
\affil[3]{Department of Physics, Lomonosov Moscow State University, Leninskie Gory, 119991 Moscow, Russia}
\date{}
\newtheorem{theorem}{Theorem}

\newtheorem{example}[theorem]{Example}

\newtheorem{proposition}[theorem]{Proposition}
\newtheorem{remark}[theorem]{Remark}

\newenvironment{proof}[1][Proof]{\noindent\textbf{#1.} }{\ \rule{0.5em}{0.5em}}

\begin{document}
\maketitle

%
%

\begin{abstract}
Steady adiabatic filtration of real gases is studied. Thermodynamical states of real gases are presented by Legendrian surfaces in 5-dimensional thermodynamical contact space. The relation between phase transitions and singularities of projection of the Legendrian surfaces on the plane of intensive variables is shown. The constructive method of finding solutions of the Dirichlet filtration problem together with analysis of critical phenomena is presented. Cases of van der Waals and Peng-Robinson gases are discussed in details.
\end{abstract}
\noindent{\it Keywords\/}: phase transitions, thermodynamics, filtration, porous media

\section{Introduction}
In this paper we continue (see~\cite{GLRT,LychGSA}) studies of phase
transitions in processes describing by nonlinear partial differential
equations. Here we consider the 3-dimensional steady adiabatic filtrations
of gases. \ The first valuable results in this area were obtained by
Leibenson L.S.~\cite{Lib} and Muskat M.~\cite{Mus}. They have proposed
a generalization of the Navier-Stokes equations for this case, where they
substituted the Newton law by the Darcy one.

The condition for stationary of filtration not only simplifies the
mathematical model but also has the practical reason because the development
of such processes not only take a long time but also  control of them  uses
the so-called \textit{cascade} \textit{method.}

The most complete results in this paper  are obtaining for the case of
filtrations in homogeneous and isotropic media, where we are able to get the
explicit formulae for solutions of the Dirichlet problem.

The paper is organized in the following way. The first part is devoted to
thermodynamics of gases and in the second part we apply its results to the
filtration problems.

We included the part devoted to the thermodynamics of gases for the
following reasons. First of all, we formulated explicitly (see also~\cite{LY}%
) that thermodynamical states are Legendrian or Lagrangian manifolds in the
corresponding contact or symplectic spaces. These manifolds equipped with a
quadratic differential forms, and domains of applicability of the
thermodynamic model are exactly domains where this form defines a Riemannian
structure. This allows us to locate domains of applicability of our model as
well as to find singularities of projections of the Lagrangian manifolds on
spaces of intensive and extensive variables. They are exactly the
submanifolds where the quadratic differential form changes its type.

Secondly, we use symplectic (and contact) geometry (see, for example,~\cite%
{KLR}) and corresponding Poisson and Lagrange brackets to complete equations
of state to the system defining Lagrangian manifold, in the cases when we
know only part of them. For the number of cases it was done by introducing and
using of the Massieu-Plank potential. We show that this potential appears in
the very natural way in an attempt to find caloric equation of state in the
case when we know a thermic equation~\cite{Fort}. We show and give
explicit formulae for expression of the main thermodynamical potentials,
heat capacities, speed of sound as well as curves where phase transitions
occur (so called coexistence curves) in terms of the Massieu-Plank potential.

To illustrate this approach we have chosen two the more popular models:
van der Waals and Peng-Robinson gases.

We use these models also to illustrate filtrations of these gases in more
details. Especially it concerns to filtration in a domain with one source
where the behaviour of gas and phase transitions are highly nontrivial.
\section{Thermodynamics of real gases}

\subsection{Preliminaries}

A thermodynamical system is described by two types of variables: extensive
and intensive (the notion was proposed by Richard Tolman in 1917~\cite{Tol}).

We postpone discussion of intensive variables and begin with extensives.

The defining property of extensive variables is their additivity with
respect to division of a system by a disjoint union of subsystems. The main
examples of extensives are: mass -- $m,$ volume -- $V,$ inner energy -- $E,$
entropy -- $S,$ etc.

Let's denote extensive variables as $\left( E,S,X\right) ,$ where $X=\left(
X_{1},...,X_{n}\right) ,$ and to each extensive (except $E$ ), i.e.
variables $\left( S,X\right) $ corresponds intensive variable $\left(
T,Y\right) ,$ where $T$ is a temperature and $Y=\left( Y_{1},..,Y_{n}\right)
.$

After such division the main law of thermodynamics (containing of the first
and the second laws) states that differential $1$-form%
\begin{equation*}
\theta =dE-TdS+YdX
\end{equation*}%
should be zero.

In other words, if we consider a $\left( 2n+3\right) $-dimensional space
$\mathbb{R}^{2n+3}$ with coordinates $\left( E,S,X,T,Y\right) $ and equipped
with differential form $\theta ,$ then a state of thermodynamical system is
a such submanifold $L\subset \mathbb{R}^{2n+3}$ where the main law of
thermodynamics holds, i.e. the restriction $\left. \theta \right\vert
_{L}=0. $ The pair $\left( \mathbb{R}^{2n+3},\theta \right) $ is the
standard model of the contact space (see for example,~\cite{KLR}) and $L$
is an integral manifold of $\theta .$ We'll require that $L$ is a maximal
integral manifold , i.e. Legendrian manifold, the notion proposed by
Vladimir Arnold, then $\dim L=n+1.$

Thus, by a \textit{thermodynamical state }we mean a Legendrian manifold in
the contact space $\left( \mathbb{R}^{2n+3},\theta \right) .$

In the case when functions $\left( S,X\right) $ are coordinates on a
Legendrian manifold $L,$ this manifold could be written in the form
\begin{equation*}
L=\left\{ E=F\left( S,X\right) ,\ T=\frac{\partial F}{\partial S},\ Y=-\frac{%
\partial F}{\partial X}\right\} ,
\end{equation*}%
and the restriction of intensive variables on $L$ plays the role of \textit{%
thermodynamical forces}.

In what follows we shall mainly consider another (but equivalent) form
\begin{equation*}
\theta ^{\prime }=-T^{-1}\theta =dS-T^{-1}dE-T^{-1}Y\ dX.
\end{equation*}%
Then the condition to be Legendrian will be not changed, but the projection $%
\pi\colon\mathbb{R}^{2n+3}\to\mathbb{R}^{2n+2},$ where $\pi \left(
E,S,X,T,Y\right) =\left( E,X,T,Y\right) ,$ allows us to eliminate $S$ from
the description of thermodynamical states. Indeed (see, for example,~\cite%
{KLR}, for more details), the restriction of the projection $\pi$ on a
Legendrian manifold $L$, $\pi\colon L\to \mathbb{R}^{2n+2},$ is an
immersion and the image $\widehat{L}=\pi \left( L\right) \subset \mathbb{R}%
^{2n+2}$ an immersed Lagrangian submanifold in the symplectic space $\left(
\mathbb{R}^{2n+2},\Omega \right) ,$ where the structure form $\Omega $ \
equals%
\begin{equation*}
\Omega =-d\theta ^{\prime }=d\left( T^{-1}\right) \wedge dE+d\left(
T^{-1}Y\right) \wedge dX.
\end{equation*}%
In the case, when functions $\left( E,X\right) $ are coordinates on $%
\widehat{L},$ this Lagrangian manifold can be written in the form
\begin{equation*}
T=\left( \frac{\partial F}{\partial E}\right) ^{-1},\ Y=\frac{\partial F}{%
\partial X}\left( \frac{\partial F}{\partial E}\right) ^{-1},
\end{equation*}%
for some function $F=F\left( E,X\right) .$

In practice, Lagrangian manifolds $\widehat{L}$ are defined by "physical
laws"%
\begin{equation*}
L=\left\{ f_{1}\left( E,X,T,Y\right) =0,...,\ f_{n+1}\left( E,X,T,Y\right)
=0\right\} ,
\end{equation*}
which are independent, i.e.
\begin{equation*}
df_{1}\wedge \cdots \wedge df_{n+1}\neq 0\quad\mbox{at points of}\quad\widehat{L},
\end{equation*}
and such that all the Poisson brackets $\left[ f_{i},f_{j}\right] =0,\
i,j=1,...,n+1,$ on $\widehat{L}.$

Here the Poisson brackets are taken with respect to structure form $\Omega :$%
\begin{equation*}
\left[ f_{i},f_{j}\right] \ \Omega ^{\wedge \left( n+1\right) }\stackrel{\mathrm{def}}{=}df_{i}\wedge df_{j}\wedge \Omega ^{\wedge n}.
\end{equation*}

Remark that not all points of the Lagrangian manifold are equally good for us.
Namely, if we consider a thermodynamics as a theory describing a
measurement of extensive variables by random processes (see~\cite{LY} for
more details), then points on the Lagrangian manifolds correspond to extreme
probability measures, finding by using Maximum Entropy Principle or
better -- Minimal Information Gain, and values of extensive are exactly the
means of these measures. The second central moment corresponds to a
quadratic differential form on $\widehat{L}.$ This form could be obtained by
the restriction of the universal quadratic form $-\kappa $ on $\mathbb{R}%
^{2n+2},$ where%
\begin{equation*}
\kappa =d\left( T^{-1}\right) \cdot dE+d\left( T^{-1}Y\right) \cdot dX,
\end{equation*}%
and dot $\cdot $ stands for symmetric product of differential $1$-forms.

This means that we have a domain $D\subset \widehat{L}$ of \textit{%
admissible points,} where quadratic differential form $\left. \kappa
\right\vert _{\widehat{L}}$ is negatively defined. We'll see later on \ that
the boundary of the domain $D$ is the set where we'll expect phase
transitions.

\subsection{Equilibrium}

Let's consider a system defined by extensive variables $\left( E,S,X\right) $
and let this system be a disjoint union of two subsystems $\left(
E_{1},S_{1},X_{1}\right) $ and $\left( E_{2},S_{2},X_{2}\right) $ and let
\begin{eqnarray*}
\theta &=&dE-TdS-YdX,\\
\theta_{1}&=&dE_{1}-T_{1}dS_{1}-Y_{1}dX_{1},\qquad
\theta_{2}=dE_{2}-T_{2}dS_{2}-Y_{2}dX_{2}
\end{eqnarray*}%
be the corresponding forms.

Then
\begin{equation*}
E=E_{1}+E_{2},\ S=S_{1}+S_{2},\ X=X_{1}+X_{2},
\end{equation*}%
and
\begin{equation*}
\theta =\theta _{1}+\theta _{2}+(T_{1}-T)dS_{1}+(T_{2}-T)\ dS_{2}+\left(
Y_{1}-Y\right) dX_{1}+\left( Y_{2}-Y\right) dX_{2}.
\end{equation*}%
On the given thermodynamical state we have $\theta =0,\theta _{1}=0$ and $%
\theta _{2}=0,$ therefore
\begin{equation*}
T_{1}=T_{2}=T,\ Y_{1}=Y_{2}=Y,
\end{equation*}%
i.e. for the union of subsystems extensive variables are added but intensive
should be equalized.

\begin{example}[Gibbs potential]
To describe gases we use extensives $\left( E,S,V,m\right) -$ inner energy,
entropy, volume and mass and in this case
\begin{equation*}
\theta =dE-TdS+pdV-\gamma dm,
\end{equation*}%
where $p$ is a pressure and $\gamma $ is a chemical potential.

Assume that a thermodynamical state is given by function $E=F\left(
S,V,m\right) .$ Then the additivity property for variables $\left(
E,S,V,m\right) $ means that $F$ is a homogeneous function of degree $1$ and
we'll write it in the form
\begin{equation*}
F\left( S,V,m\right) =mF\left( \frac{S}{m},\frac{V}{m},1\right) ,
\end{equation*}%
or in terms of specific energy $\varepsilon =E/m,\ $specific entropy $\sigma
=S/m,$ and specific volume $v=V/m$ as
\begin{equation}
\varepsilon =f\left( \sigma ,v\right) .  \label{Gst1}
\end{equation}%
Then
\begin{equation*}
\theta =m(d\varepsilon -Td\sigma +pdv)+\left( \varepsilon -T\sigma
+pv-\gamma \right) dm,
\end{equation*}%
and for state equation (\ref{Gst1}) we have
\begin{equation*}
d\varepsilon -Td\sigma +pdv=0,
\end{equation*}%
and
\begin{equation*}
\gamma =\varepsilon +pv-T\sigma .
\end{equation*}%
Function $G=E+pV-TS$ is called Gibbs free energy, and therefore $\gamma $ is
a specific Gibbs free energy.

The discussed above conditions of equilibrium for gases shall take now the
form:%
\begin{equation}
T_{1}=T_{2},\ P_{1}=P_{2},\ \gamma _{1}=\gamma _{2}.  \label{GSequi}
\end{equation}
\end{example}

\subsection{State equations for gases}

Consider a 5-dimensional contact space $\mathbb{R}^{5}$ with coordinates $%
\left( \sigma ,v,\varepsilon ,p,T\right) ,$where $\varepsilon ,\sigma ,v$
stand for the specific inner energy, specific entropy and specific volume
respectively, $p$ is the pressure and $T$ is the temperature.

In this case, the above contact structure is given by differential 1-form%
\begin{equation*}
\theta =d\sigma -T^{-1}d\varepsilon -T^{-1}pdv,
\end{equation*}%
and by thermodynamical states we mean Legendrian surfaces $L\subset
\mathbb{R}^{5},$ i.e. surfaces where $\left. \theta \right\vert _{L}=0,$ or,
in other words, where  the law of energy conservation holds.

As above to eliminate specific entropy $\sigma $ from our consideration we
consider a 4-dimensional space $\mathbb{R}^{4}$ with coordinates $\left(
v,\varepsilon ,p,T\right) $ and projection%
\begin{equation*}
\pi\colon\mathbb{R}^{5}\to\mathbb{R}^{4}\mathbf{,}
\end{equation*}%
where $\pi \left( \sigma ,v,\varepsilon ,p,T\right) =\left( v,\varepsilon
,p,T\right) .$

The restriction of projection $\pi $ on $L$ is an immersion. In order to
avoid extra technicalities, we'll assume that $\pi\colon L\to\widehat{L}$
is a diffeomorphism of $L$ with a surface $\widehat{L}\subset
\mathbb{R}^{4}\mathbf{.}$

The last is a Lagrangian surface in 4-dimensional symplectic space $\mathbb{R%
}^{4}$ \ equipped with the symplectic 2-form%
\begin{eqnarray*}
\Omega=-d\theta &=&d\left( T^{-1}\right) \wedge d\varepsilon +d\left( pT^{-1}\right)
\wedge dv= \\
&=&T^{-1}dp\wedge dv-T^{-2}dT\wedge \left( d\varepsilon +pdv\right) ,
\end{eqnarray*}%
i.e. $\left. \Omega \right\vert _{\widehat{L}}=0.$

Remark, that conditions $\left. \theta \right\vert _{L}=0$ and $%
\left. \Omega \right\vert _{\widehat{L}}=0$ are equivalent up to a shift of $\ L$
along axis $\sigma $ and moreover, \ 2-form $d\theta $ is the pullback of $%
\Omega $, $\pi ^{\ast }\left( \Omega \right) =d\theta .$

The Lagrangian surface $\widehat{L}$ is defined by two equations (or by two
thermodynamical laws)
\begin{equation*}
f\left( v,\varepsilon ,p,T\right) =0,\quad g\left( v,\varepsilon ,p,T\right) =0,
\end{equation*}
where functions $f$ and $g$ are independent and the condition $\left. \Omega
\right\vert _{\widehat{L}}=0$ is equivalent to the condition that the Poisson bracket $%
[f,g]$ vanishes on $\widehat{L}.$

In our case this bracket could be defined by the relation%
\begin{equation*}
df\wedge dg\wedge \Omega =[f,g]\ \Omega \wedge \Omega ,
\end{equation*}%
and has the following form%
\begin{equation*}
\lbrack f,g]=\frac{1}{2}\left(pT\left( f_{p}g_{\varepsilon }-f_{\varepsilon
}g_{p}\right) +T^{2}\left( f_{T}g_{\varepsilon }-f_{\varepsilon
}g_{T}\right) +T\left( f_{v}g_{p}-f_{p}g_{v}\right)\right).
\end{equation*}

To find state manifolds $\widehat{L}$ for real gases we'll  take the following two
equations:

\begin{enumerate}
\item \textit{Thermic equation of state}%
\begin{equation*}
p=A\left( v,T\right) ,
\end{equation*}%
and

\item \textit{Caloric equation of state}
\begin{equation*}
\varepsilon =B\left( v,T\right) .
\end{equation*}
\end{enumerate}

Then the compatibility condition for them
\begin{equation*}
[ p-A\left( v,T\right) ,\varepsilon -B\left( v,T\right) ]=0\quad\mbox{if}\quad\left\{ p=A\left( v,T\right) ,\varepsilon =B\left(
v,T\right) \right\}
\end{equation*}%
is equivalent to relation
\begin{equation*}
\left( T^{-2}B\right) _{v}=\left( T^{-1}A\right) _{T}.
\end{equation*}

Solutions of the last equation we'll present by using potential function $%
\phi \left( v,T\right) $ such that $T^{-1}A=R\phi _{v}$ and $T^{-2}B=R\phi
_{T},$ or
\begin{equation}
A=RT\phi _{v},\ \ B=RT^{2}\phi _{T},  \label{fpotential}
\end{equation}%
where $R$ is the universal gas constant.

In 1901 H. Kamerlingh Onnes proposed a virial model
(viris=force, Latin) for thermic equation~\cite{Onnes}:
\begin{equation*}
A=\frac{RT}{v}Z\left( v,T\right) ,
\end{equation*}%
where $Z\left( v,T\right) $ is a \textit{compressibility factor }of the form%
\begin{equation*}
Z\left( v,T\right) =1+\frac{A_{1}\left( T\right) }{v}+\cdots +\frac{%
A_{k}\left( T\right) }{v^{k}}+\cdots ,
\end{equation*}%
and $A_{k}\left( T\right) $ are \textit{virial coefficients}.

We have  the following relation between potential function and the
compressibility factor:
\begin{equation*}
Z=v\phi _{v}.
\end{equation*}%
Therefore,
\begin{equation*}
\phi =\alpha \left( T\right) +\ln v-A_{1}\left( T\right) v^{-1}-\frac{1}{2}%
A_{2}\left( T\right) v^{-2}-\cdots -\frac{1}{k}A_{k}\left( T\right)
v^{-k}-\cdots
\end{equation*}%
where $\alpha \left( T\right) $ is a smooth function.

Remark that in the case of ideal gases we have trivial virial coefficients, $%
A_{i}=0,$ and
\begin{equation*}
B=\frac{n}{2}RT,
\end{equation*}%
where $n$ is the degree of freedom .

This means that
\begin{equation*}
T^{2}\alpha _{T}=\frac{n}{2}T,
\end{equation*}%
and therefore the potential function $\phi $ has to be in the form
\begin{equation}
\phi =\frac{n}{2}\ln T+\ln v-A_{1}\left( T\right) v^{-1}-\frac{1}{2}%
A_{2}\left( T\right) v^{-2}-\cdots -\frac{1}{k}A_{k}\left( T\right)
v^{-k}-\cdots .  \label{MPpotential}
\end{equation}

In order to understand the meaning of potential function $\phi $ \ let's
express the main thermodynamical functions in terms of $\phi .$

Let's start with the specific entropy. Condition $\left. \theta \right\vert
_{L}=0$ gives us
\begin{eqnarray*}
\sigma _{T} &=&R\left( 2\phi _{T}+T\phi _{TT}\right) , \\
\sigma _{v} &=&R\left( \phi _{v}+T\phi _{vT}\right) ,
\end{eqnarray*}%
and therefore (up to a constant)%
\begin{equation}
\sigma =R\left( \phi +T\phi _{T}\right) .  \label{entropy}
\end{equation}

For the specific Gibbs free energy $\gamma $ we get%
\begin{equation}
\gamma =\varepsilon +pv-T\sigma =RT\left( v\phi _{v}-\phi \right) ,
\label{Gibbs}
\end{equation}%
and for specific enthalpy we have%
\begin{equation*}
\eta =\varepsilon +pv=RT\left( T\phi _{T}+v\phi _{v}\right).
\end{equation*}

Similarly, for the Helmholtz free energy or Massieu-Plank potential we get%
\begin{equation}
\Xi =\sigma -\frac{\varepsilon }{T}=R\phi .  \label{Mpotential}
\end{equation}%
For this reason, from now and on, we'll call $\phi $ \textit{Massieu-Plank
potential}.

\subsection{Riemannian structures on Lagrangian manifolds}
As we have noted above, not all points of the Lagrange surface $\widehat{L}$
correspond to thermodynamical states.

Namely, the fundamental quadratic differential form $\kappa $ on the surface
defines the domain of applicable states, i.e. states where $\left. \kappa
\right\vert _{\widehat{L}}$ is negative~\cite{LY}:

In our case, this form could be written in the following way%
\begin{equation*}
\kappa =d\left( T^{-1}\right) \cdot d\varepsilon +d\left( pT^{-1}\right)
\cdot dv,
\end{equation*}%
and in terms of Massieu-Plank potential it has the following form%
\begin{equation}
R^{-1}\kappa =-\left( \phi _{TT}+2T^{-1}\phi _{T}\right) dT\cdot dT+\phi
_{vv}dv\cdot dv.  \label{kappa}
\end{equation}

\begin{remark}
Let's assume that the above quadratic differential form is non degenerated.
Then $\kappa $ is a flat metric if and only if $\phi \left( T,v\right) =%
\frac{n}{2}\ln T+\phi _{0}\left( v\right) ,$ or if all virial coefficients
are constants.
\end{remark}

Summarizing this discussion we get the following description of the real gas
states.

\begin{theorem}
Thermodynamical states of real gases are defined by \textit{Massieu-Plank
potential function }$\phi \left( T,v\right) $ and have the following form:%
\begin{equation*}
p=RT\phi _{v},\quad\varepsilon =RT^{2}\phi _{T},\quad\sigma =R\left( \phi +T\phi
_{T}\right) ,
\end{equation*}%
where function $\phi $ has the following expression in terms of virial
coefficients%
\begin{equation*}
\phi =\frac{n}{2}\ln T+\ln v-A_{1}\left( T\right) v^{-1}-\frac{1}{2}%
A_{2}\left( T\right) v^{-2}-\cdots -\frac{1}{k}A_{k}\left( T\right)
v^{-k}-\cdots .
\end{equation*}%
The domain of applicable states on the plane $\left( T,v\right) $ is given
by inequalities%
\begin{equation}
\phi _{vv}<0,\ T\phi _{TT}+2\phi _{T}>0.  \label{applicable}
\end{equation}%
Phase transitions occur near the curve%
\begin{equation*}
\phi _{vv}=0\hspace{3mm}or\hspace{3mm}T\phi _{TT}+2\phi _{T}=0.
\end{equation*}
\end{theorem}

\begin{remark}
Due to the state equations we have $\varepsilon _{T}=RT\left( T\phi
_{TT}+2\phi _{T}\right) $ and $p_{v}=RT\phi _{vv}.$ Therefore, applicable
states belong to domain where
\begin{equation*}
\varepsilon _{T}>0\hspace{3mm}or\hspace{3mm}p_{v}<0,
\end{equation*}%
or, because we expect that $\varepsilon _{T}>0,$ the phase transitions occur
near the curve
\begin{equation*}
\phi _{vv}=0.
\end{equation*}
\end{remark}

\subsection{Singularities and phase transitions}

Denote by $L_{\phi }\subset \mathbb{R}^{4}\left( \varepsilon ,v,T,p\right) $
the Lagrangian surface that corresponds to the \textit{Massieu-Plank
function }$\phi \left( T,v\right) .$ This surface is given by equations%
\begin{equation}
\label{eqforPhi}
p =RT\phi _{v}, \quad
\varepsilon =RT^{2}\phi _{T}.
\end{equation}%
The projection $L_{\phi }\rightarrow \mathbb{R}^{2}\left( T,p\right) $ of
this surface on the plane $\left( T,p\right) $ of intensive variables has
singularities at points $\Sigma _{i}\subset L_{\phi },$ where differential $2
$-form $dp\wedge dT$ equals zero, or
\begin{equation*}
\Sigma _{i}=\left\{ \phi _{vv}=0\right\} .
\end{equation*}%
In the similar way, singularities $\Sigma _{e}$ of the projection $L_{\phi
}\rightarrow \mathbb{R}^{2}\left( \varepsilon ,v\right) $ on the plane of
extensive variables correspond to the points where differential $2$-form $%
d\varepsilon \wedge dv$ has zeroes, i.e.%
\begin{equation*}
\Sigma _{e}=\left\{ T\phi _{TT}+2\phi _{T}=0\right\} .
\end{equation*}%
In other words, we have singularities of two types that correspond to
singularities $\Sigma _{i}$ or $\Sigma _{e}.$ In the first case, the jumps
preserve values of intensive variables and have discontinuous for extensive
ones, and this is exactly what we call \textit{phase transitions}, but in
the second type singularities the jumps preserve extensive and we observe
discontinuous for intensive variables. Remark that both these singularities are
essential for us, because quadratic differential form $\kappa $ changes its
type at these points.

We say that distinct points $\left( v_{1},p_{1},\varepsilon
_{1},T_{1}\right) $ and $\left( v_{2},p_{2},\varepsilon _{2},T_{2}\right) $
on the Lagrangian surface $L_{\phi }$ are \textit{phase equivalent }iff%
\begin{equation*}
p_{1}=p_{2},\quad T_{1}=T_{2},\quad\gamma \left( v_{1},T_{1}\right) =\gamma \left(
v_{2},T_{2}\right) .
\end{equation*}

\begin{remark}
We have $\gamma =\varepsilon +pv-T\sigma ,$ and
\begin{equation*}
d\gamma =vdp-\sigma dT\ \mathrm{mod}\left( \theta \right) .
\end{equation*}%
This relation explains the Maxwell rule of equal areas on the plane $\left(
v,p\right) $ (or  $\left( \sigma ,T\right) )$ that used to apply to find
phase equivalent points.
\end{remark}

In terms of \ the Massieu-Plank potential phase equivalent points $\left(
v_{1},T\right) $ and $\left( v_{2},T\right) $ \ on surface $L_{\phi }$ could
be found from the following system of equations:
\begin{equation}
\label{PhaseEqui}
\begin{split}
&\phi _{v}\left( v_{2},T\right) -\phi _{v}\left( v_{1},T\right)=0,{}\\&
\phi \left( v_{2},T\right) -\phi \left( v_{1},T\right) -v_{2}\phi _{v}\left(
v_{2},T\right) +v_{1}\phi _{v}\left( v_{1},T\right)=0.
\end{split}
\end{equation}

Eliminating $v_{2}$ from the above equations and putting $v_{1}=v$ we get phase
transition or \textit{coexistence} curve $\Gamma _{\phi }\subset \mathbb{R}^{2}\left( v,T\right) $
which shows $\left( v,T\right) $ points of phase transitions, and eliminating $%
T$ we get phase transition curve $\Gamma _{\phi }\subset \mathbb{R}%
^{2}\left( v_{1},v_{2}\right) $ shows the specific volumes of phase
transitions.

Moreover, if we consider the equivalent equations
\begin{equation}
\begin{split}
\label{PhaseEqui1}
&\phi _{v}\left( v_{2},T\right)=\frac{p}{RT},\ \ \phi _{v}\left(
v_{1},T\right) =\frac{p}{RT},{}\\&
\phi \left( v_{2},T\right) -\phi \left( v_{1},T\right) -v_{2}\phi _{v}\left(
v_{2},T\right)+v_{1}\phi _{v}\left( v_{1},T\right) =0,
\end{split}
\end{equation}
and eliminate $\left( v_{1},v_{2}\right) $ from these equations we get phase
transition curve $\Gamma _{\phi }\subset \mathbb{R}^{2}\left( p,T\right) $
that shows the pressures and temperatures when phase transitions occur.

These equations also allow us to find jumps in heat, inner energy and work
on phase transitions:%
\begin{eqnarray*}
\Delta Q &=&T\left( S\left( v_{2},T\right) -S\left( v_{1},T\right) \right)
=RT\left( \Delta \phi +T\Delta \phi _{T}\right) , \\
\Delta W &=&-P\left( v_{2}-v_{1}\right) =-RT~\phi _{v}\ \Delta v=-RT\ \Delta
\phi ~, \\
\Delta \varepsilon  &=&\varepsilon \left( v_{2},T\right) -\varepsilon \left(
v_{1},T\right) =RT^{2}\Delta \phi _{T}=\Delta Q+\Delta W.
\end{eqnarray*}

\subsection{Heat capacities and speed of sound}

In this section we compute heat capacities and speed of sound in terms of
potential function $\phi .$

Recall that \textit{\ heat capacity} measures the amount of heat required to
change temperature by a given amount. There are two types of heat
capacities: $C_{v}$ -- heat capacity at fixed volume, $C_{p}$ -- heat capacity at
fixed pressure.

We begin with computation of $C_{v}.$ Following to the above saying, we define
it by the relation%
\begin{equation*}
Td\sigma -C_{v}dT\equiv 0\quad\mathrm{mod}\left\langle dv\right\rangle ,
\end{equation*}%
on $L_{\phi }.$

Using the description of the Lagrangian surface $L_{\phi }$ we get
\begin{equation*}
d\sigma -R\left( 2\phi _{T}+T\phi _{TT}\right) dT\equiv 0\quad\mathrm{mod}%
\left\langle dv\right\rangle ,
\end{equation*}%
and therefore%
\begin{equation*}
C_{v}=RT\left( 2\phi _{T}+T\phi _{TT}\right) =\varepsilon _{T}.
\end{equation*}

In the similar way we define $C_{p}$ by the following relation:
\begin{equation*}
Td\sigma -C_{p}dT\equiv 0\quad\mathrm{mod}\left\langle dp\right\rangle .
\end{equation*}%
We have
\begin{eqnarray}
dp=R\left(T\phi_{vv}dv+(\phi_{v}+T\phi_{Tv})dT\right)\label{ps relation} \\
Td\sigma=RT\left( \left( 2\phi _{T}+T\phi _{TT}\right) dT+\left( \phi
_{v}+T\phi _{vT}\right) dv\right) .
\end{eqnarray}%
Therefore,%
\begin{equation*}
dv\equiv -\frac{\phi _{v}+T\phi_{Tv}}{T\phi _{vv}}dT\quad\mathrm{mod}%
\left\langle dp\right\rangle
\end{equation*}%
and
\begin{equation*}
C_{p}=\frac{R}{\phi _{vv}}\left( T^{2}\left( \phi _{TT}\phi _{vv}-\phi
_{Tv}^{2}\right) +2T\left( \phi _{T}\phi _{vv}-\phi _{v}\phi _{Tv}\right)
-\phi _{v}^{2}\right) =\varepsilon _{T}-\frac{p_{T}^{2}T}{p_{v}},
\end{equation*}%
also%
\begin{equation*}
C_{p}-C_{v}+\frac{p_{T}^{2}T}{p_{v}}=0.
\end{equation*}

The speed of sound $c$ is defined as $c^{2}=C_{s},$ where
\begin{equation*}
dp-C_{s}dv^{-1}\equiv 0\quad\mathrm{mod}\left\langle d\sigma \right\rangle .
\end{equation*}%
In our case,%
\begin{equation*}
dT\equiv -\frac{\varepsilon _{v}+p}{\varepsilon _{T}}dv\quad\mathrm{mod}%
\left\langle d\sigma \right\rangle ,
\end{equation*}%
and therefore
\begin{eqnarray}
C_{s} &=&v^{2}\left( \frac{\varepsilon _{v}p_{T}-\varepsilon _{T}p_{v}}{%
\varepsilon _{T}}+\frac{1}{2}\frac{\left( p^{2}\right) _{T}}{\varepsilon _{T}%
}\right) =  \label{Cs formula} \\
&=&Rv^{2}\frac{T^{2}\left( \phi _{Tv}^{2}-\phi _{TT}\phi _{vv}\right)
+2T\left( \phi _{v}\phi _{vT}-\phi _{T}\phi _{vv}\right) +\phi _{v}^{2}}{%
2\phi _{T}+T\phi _{TT}}.
\end{eqnarray}%
It is easy to check that all the quantities $C_{v},C_{p}$ and $C_{s}$ are
connected by the following relation%
\begin{equation*}
T^{-1}C_{s}C_{v}+Rv^{2}\phi _{vv}C_{p}=0.
\end{equation*}

As a by-product of this equality we get the following observation.

\begin{theorem}
The speed of sound vanishes on singular set $\Sigma _{i}.$
\end{theorem}

\subsection{Relations with Monge-Ampere equations}
Formula~(\ref{Cs formula}) for $C_{s}$ gives us a method to find equations of state
for a media with known speed of sound:%
\begin{equation*}
c^{2}=R\ F\left( v,T\right) v^{2}.
\end{equation*}

Indeed, in this case the formula shows that the Massieu-Plank potential%
\textit{\  }$\phi \left( T,v\right) $ satisfies the following Monge-Ampere
differential equation%
\begin{eqnarray*}
&&T^{2}\left( \phi _{Tv}^{2}-\phi _{TT}\phi _{vv}\right) +2T\left( \phi
_{v}\phi _{vT}-\phi _{T}\phi _{vv}\right) +\phi _{v}^{2}- \\
&&-F\left( v,T\right) ~\left( 2\phi _{T}+T\phi _{TT}\right) =0.
\end{eqnarray*}

Thus, for example, equation of state for ultrarelativistic fluids has the
form
\begin{equation*}
pv=c^{2}
\end{equation*}%
and the last equation takes the form%
\begin{equation*}
vT\left( \phi _{Tv}^{2}-\phi _{TT}\phi _{vv}\right) +2v\left( \phi _{v}\phi
_{vT}-\phi _{T}\phi _{vv}\right) +vT^{-1}\phi _{v}^{2}-\phi _{v}\left( 2\phi
_{T}+T\phi _{TT}\right) =0.
\end{equation*}

Also the virial expansion could be used to find approximations to solutions
for  the obtained Monge-Ampere equations.

Thus, for main thermodynamical functions we have%
\begin{eqnarray*}
\varepsilon  &=&RT\left( \frac{n}{2}-\frac{TA_{1}^{\prime }}{v}-\frac{%
TA_{2}^{\prime }}{2v^{2}}\right) +O\left( \frac{1}{v^{3}}\right) , \\
p &=&RT\left( \frac{1}{v}+\frac{A_{1}}{v^{2}}\right) +O\left( \frac{1}{v^{3}}%
\right) , \\
C_{v} &=&\frac{Rn}{2}-\frac{RT}{v}\left( TA_{1}^{\prime \prime
}+2A_{1}^{\prime }\right) -\frac{RT}{2v^{2}}\left( TA_{2}^{\prime \prime
}+2A_{2}^{\prime }\right) +O\left( \frac{1}{v^{3}}\right)  \\
C_{p} &=&\frac{R(n+2)}{2}-\frac{RT^{2}A_{1}^{\prime\prime}}{v}+O\left(\frac{1}{v^{2}}\right), \\
C_{s} &=&\frac{RT(n+2)}{n}+\frac{4RT}{vn^{2}}\left(T^{2}A_{1}^{\prime\prime}+(n+2)\left(TA_{1}^{\prime}+\frac{n}{2}A_{1}\right)\right)+O\left(\frac{1}{v^{2}}\right).
\end{eqnarray*}

\subsection{Examples}

\subsubsection{Ideal gas.}

The equations of state are the following

\begin{itemize}
\item Clapeyron-Mendeleev equation
\begin{equation*}
pv=RT,
\end{equation*}%
and

\item equation for internal energy%
\begin{equation*}
\varepsilon =\frac{n}{2}RT,
\end{equation*}%
where $n$ is the degree of freedom.
\end{itemize}

The potential function $\phi $ for this case we'll find from equations (\ref%
{eqforPhi}). We get%
\begin{eqnarray*}
\phi &=&\ln \left( T^{n/2}v\right) ,\ \sigma =R\ln \left( T^{n/2}v\right) , \\
C_{v} &=&\frac{Rn}{2},\ C_{p}=\frac{Rn}{2}+R,\ C_{s}=RT\frac{n+2}{2}.
\end{eqnarray*}%
The quadratic differential form equals%
\begin{equation*}
\kappa =-\frac{Rn}{2}\frac{dT^{2}}{T^{2}}-R\frac{dv^{2}}{v^{2}}.
\end{equation*}%
It is negative, and there are no phase transitions.

\subsubsection{van der Waals gas.}

In this case state equations are%
\begin{equation*}
p=\frac{RTv^{2}-a\left( v-b\right) }{v^{2}\left( v-b\right) },\ \varepsilon =%
\frac{n}{2}RT-\frac{a}{v}.
\end{equation*}%
Let us introduce following scale contact transformation%
\begin{equation*}
T=\frac{T}{T_{\mathrm{crit}}},\quad v=\frac{v}{v_{\mathrm{crit}}},\quad p=\frac{p}{p_{\mathrm{crit}}},\quad e=\frac{e}{e_{\mathrm{crit}}},\quad\sigma=\frac{\sigma}{\sigma_{\mathrm{crit}}},
\end{equation*}
where $T_{\mathrm{crit}}$, $v_{\mathrm{crit}}$, $p_{\mathrm{crit}}$, $e_{\mathrm{crit}}$, $\sigma_{\mathrm{crit}}$ are critical parameters for van der Waals gases:
\begin{equation*}
T_{\mathrm{crit}}=\frac{8a}{27Rb},\quad v_{\mathrm{crit}}=3b,\quad p_{\mathrm{crit}}=\frac{a}{27b^{2}},\quad e_{\mathrm{crit}}=\frac{a}{9b}, \quad \sigma_{\mathrm{crit}}=\frac{3R}{8},
\end{equation*}
then, we get the reduced equations of state in new dimensionless coordinates, which we shall continue denoting by $p$, $T$, $e$, $v$:%
\begin{equation*}
p=\frac{8T}{3v-1}-\frac{3}{v^{2}},\qquad
\varepsilon=\frac{4n}{3}T-\frac{3}{v}.
\end{equation*}%
One may show that the Massieu-Plank potential and the specific entropy for van der Waals gases are of the form:
\begin{equation}
\label{phiandent}
\phi =\ln\left(T^{n/2}(3v-1)\right)+\frac{9}{8vT},\quad\sigma =\ln \left( T^{4n/3}\left(3v-1\right)^{8/3} \right).
\end{equation}

Therefore, heat capacities for van der Waals gases have the following form:%
\begin{equation*}
C_{v} =\frac{n}{2}R, \qquad
C_{p} =\frac{C_{p,n}}{C_{p,d}}R,
\end{equation*}%
where%
\begin{eqnarray*}
C_{p,n} &=&4T(n+2)v^{3}-9nv^{2}+6nv-n,\\
C_{p,d} &=&8Tv^{3}-18v^{2}+12v-2,
\end{eqnarray*}%
and
\begin{equation*}
C_{s}=\frac{2aC_{p,n}}{3bnv(3v-1)^{2}}.
\end{equation*}

The corresponding quadratic differential form equals%
\begin{equation*}
\kappa =-\frac{Rn}{2}\frac{dT^{2}}{T^{2}}-\frac{9R(4Tv^{3}-9v^{2}+6v-1)}{4Tv^{3}(3v-1)^{2}}dv^{2}.
\end{equation*}%
Therefore, the domain of applicable states is given by the inequality%
\begin{equation*}
4Tv^{3}-9v^{2}+6v-1>0,\quad\mbox{or}\quad T>\frac{(3v-1)^{2}}{4v^{3}},
\end{equation*}%
and the phase transitions occur near the curve%
\begin{equation*}
T=\frac{(3v-1)^{2}}{4v^{3}}.
\end{equation*}

\begin{figure}[h!]
\centering
\includegraphics[scale=.4]{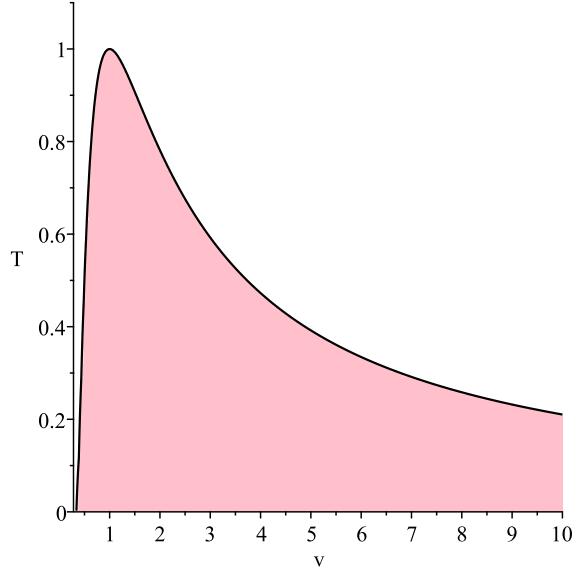}
\caption{Applicable domain for van der Waals gases.}
\label{vdWcurve}
\end{figure}

White domain in figure~\ref{vdWcurve} corresponds to the domain of applicable states.

Recall, that the \textit{coexistence curve} or \textit{binodal curve} may be obtained by means of Massieu-Plank potential using formulae~(\ref{PhaseEqui}) and~(\ref{PhaseEqui1}). Both these forms of equations of coexistence curves are essential for us, because they allow us to get such curves in different coordinates.

System~(\ref{PhaseEqui}) for van der Waals gases has the following form:
\begin{eqnarray}
\label{CoexCurve1}
\frac{(3v_{1}-1)(3v_{2}-1)(v_{1}+v_{2})}{v_{1}-v_{2}}\ln\left(\frac{3v_{2}-1}{3v_{1}-1}\right)=3(v_{1}+v_{2}-6v_{1}v_{2}),\\
\label{CoexCurve2}
T=\frac{(v_{1}+v_{2})(3v_{1}-1)(3v_{2}-1)}{8v_{1}^{2}v_{2}^{2}}.
\end{eqnarray}
Equation~(\ref{CoexCurve1}) defines a curve $\Gamma_{\phi}\subset\mathbb{R}^{2}(v_{1},v_{2})$ of specific volumes of phase transition, which is shown in figure~\ref{curvevol}, and equation~(\ref{CoexCurve2}) allows to find the corresponding temperature.

\begin{figure}
\centering
\includegraphics[scale=.3]{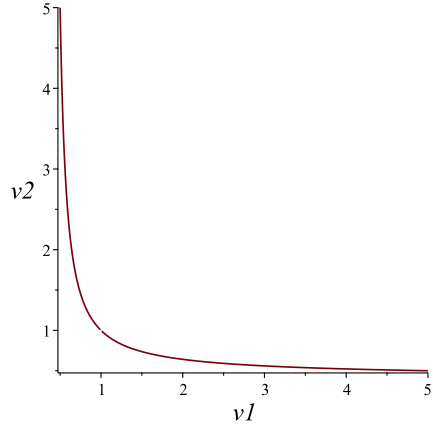}\\
\caption{Coexistence curve $\Gamma_{\phi}\subset\mathbb{R}^{2}(v_{1},v_{2})$ for van der Waals gases.}
\label{curvevol}
\end{figure}

Eliminating $v_{1}$ and $v_{2}$ from~(\ref{PhaseEqui1}) we get a binodal curve  $\Gamma_{\phi}\subset\mathbb{R}^{2}(p,T)$, which is presented in figure~\ref{curvepT}, and eliminating $v_{1}$ and putting $v_{2}=v$ we get a curve  $\Gamma_{\phi}\subset\mathbb{R}^{2}(T,v)$. It is presented in figure~\ref{curvevT}.  Since it is a difficulty to eliminate specific volumes, these curves have been obtained numerically.

\begin{figure}[ht!]
\centering \subfigure[]{
\includegraphics[width=0.4\linewidth]{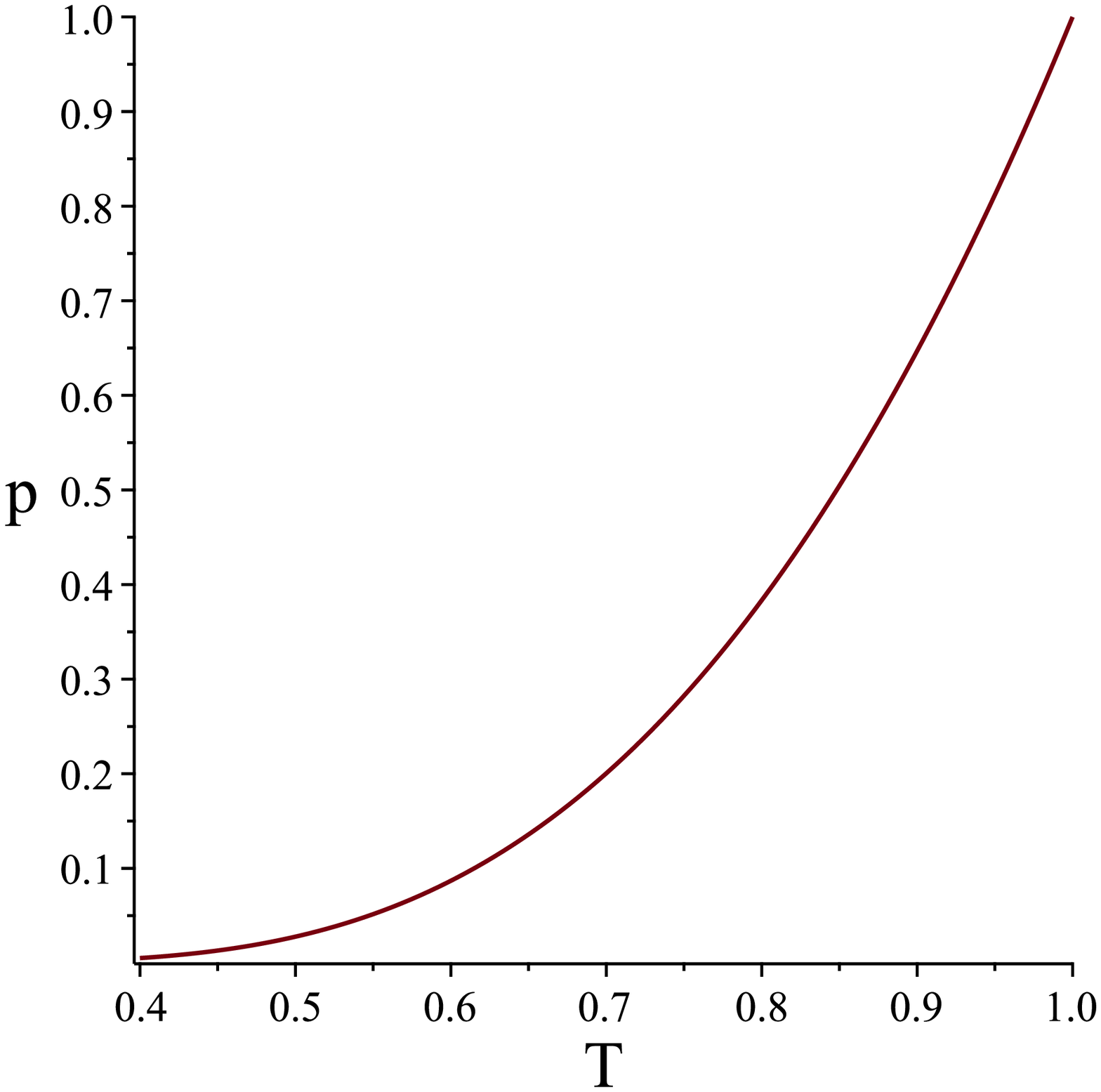} \label{curvepT} }
\hspace{4ex}
\subfigure[]{ \includegraphics[width=0.4\linewidth]{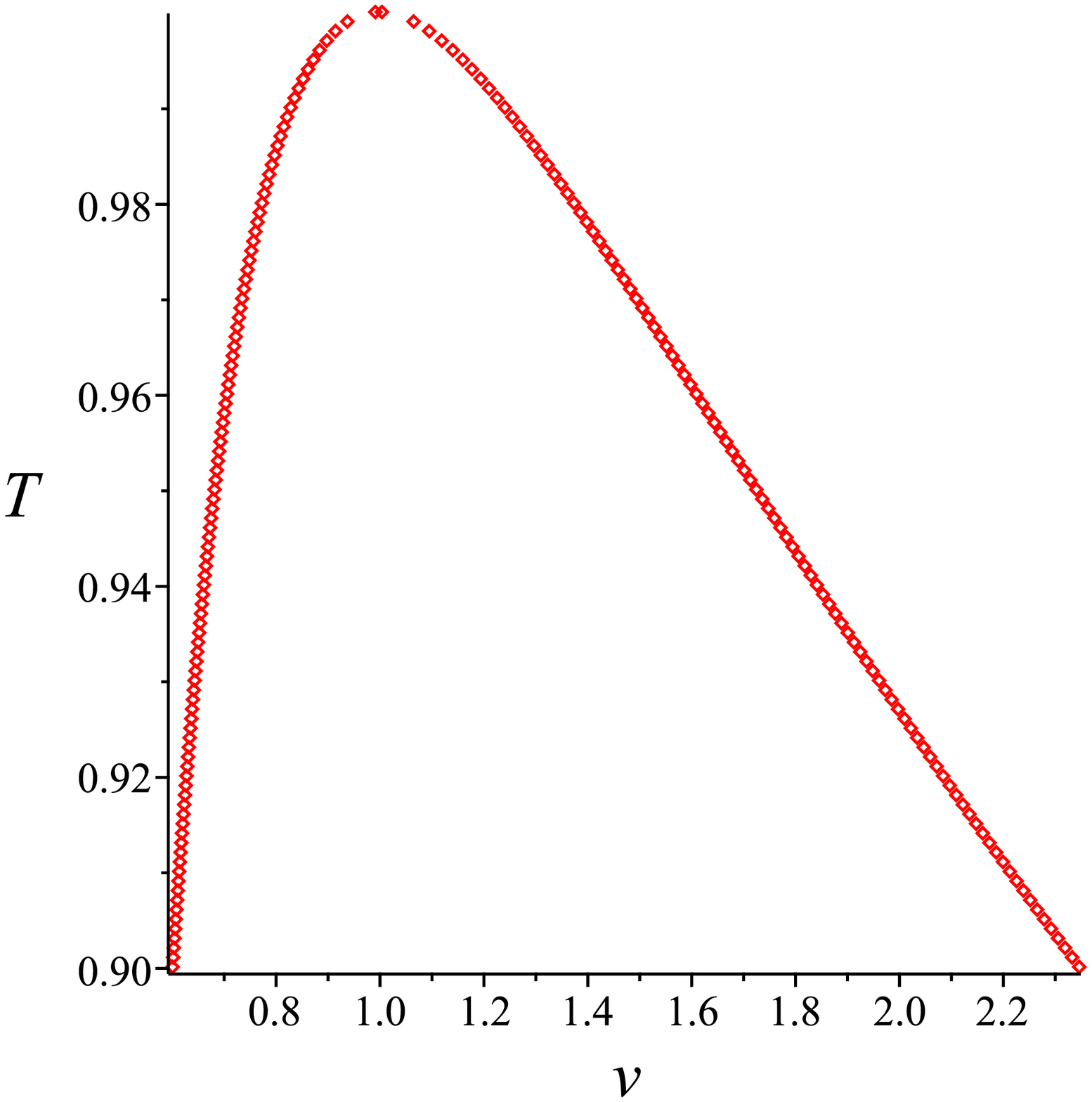} \label{curvevT} }
\caption{Coexistence curves for van der Waals gases in different coordinates: \subref{curvepT} in $\mathbb{R}^{2}(p,T)$, gas phase is on the right of the curve, liquid phase is on the left; \subref{curvevT} in $\mathbb{R}^{2}(T,v)$, inside the curve --- intermediate state (condensation process).} \label{CoexCurves}
\end{figure}

Moreover, it is possible to lift these curves into the space $\mathbb{R}^{3}(p,v,T)$, which is shown in figure~\ref{curvepvT}.

\begin{figure}
\centering
\includegraphics[scale=.4]{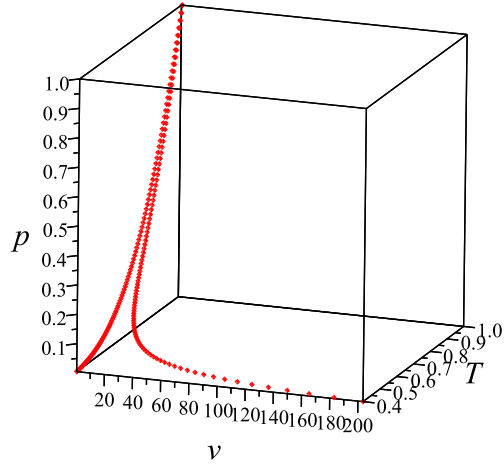}\\
\caption{Coexistence curve $\Gamma_{\phi}\subset\mathbb{R}^{3}(p,v,T)$ for van der Waals gases.}
\label{curvepvT}
\end{figure}

\subsubsection{Peng-Robinson gas.}
The Peng-Robinson EOS is the most popular equation of state for real
gases in the petroleum industry.

The first equation of state
\begin{equation*}
p=\frac{RT}{v-b}-\frac{\alpha }{\left( v+b\right) ^{2}-2b^{2}}
\end{equation*}%
was found by D.B. Robinson and his PhD student D. Peng~\cite{PR} in 1976.

To find the second equation $\varepsilon =B\left( v,T\right) ,$ we, as
above, shall take the Poisson bracket $[p-A\left( v,T\right) ,\varepsilon
-B\left( v,T\right) ],$ where
\begin{equation*}
A=\frac{RT}{v-b}-\frac{\alpha }{\left( v+b\right) ^{2}-2b^{2}},
\end{equation*}%
and require that this bracket equals zero if $p=A,\varepsilon =B.$

This condition gives differential equation on function $B:$%
\begin{equation*}
\left( b^{2}-2bv-v^{2}\right) B_{v}+\alpha =0,
\end{equation*}%
with solution%
\begin{equation*}
B=F\left( T\right) -\frac{\alpha }{\sqrt{2}b}\mathrm{arctanh}\left( \frac{v+b}{%
\sqrt{2}b}\right),
\end{equation*}%
where $F(T)$ is an arbitrary smooth function.

Using the above virial coefficients arguments we take
\begin{equation*}
\varepsilon =\frac{n}{2}RT-\frac{\alpha }{\sqrt{2}b}\mathrm{arctanh}\left(
\frac{v+b}{\sqrt{2}b}\right)
\end{equation*}%
as the second state equation.

By the following scale contact transformation%
\begin{equation*}
T\rightarrow \frac{\alpha }{bR}T,\quad\varepsilon \rightarrow \frac{\alpha }{b}%
\varepsilon,\quad p\rightarrow \frac{\alpha }{b^{2}}p,\quad v\rightarrow bv,\quad\sigma
\rightarrow R\sigma
\end{equation*}%
we transform the equations of state to the reduced form:%
\begin{eqnarray*}
p &=&\frac{T}{v-1}-\frac{1}{\left( v+1\right) ^{2}-2}, \\
\varepsilon &=&\frac{n}{2}T-\frac{\mathrm{arctanh}\left( \frac{v+1}{\sqrt{2}}%
\right) }{\sqrt{2}}.
\end{eqnarray*}%
It is easy to check that the following function
\begin{equation*}
\phi =\ln \left( T^{\frac{n}{2}}\left( v-1\right) \right) +T^{-1}\frac{\mathrm{%
arctanh}\left( \frac{v+1}{\sqrt{2}}\right) }{\sqrt{2}},
\end{equation*}%
is the Massieu-Plank potential for the Peng-Robinson gases, and%
\begin{equation*}
\sigma =\ln \left( T^{\frac{n}{2}}\left( v-1\right) \right) .
\end{equation*}

Therefore, we have for the Peng-Robinson gases%
\begin{eqnarray*}
C_{v} &=&\frac{n}{2}R, \\
C_{p} &=&\frac{C_{p,n}}{C_{p,d}}R,
\end{eqnarray*}%
where%
\begin{eqnarray*}
C_{p,n} &=&T(n+2)v^4+2\left(2(n+2)T-n\right)v^3+2\left((n+2)T+n\right)v^2-\\&&-2\left(2(n+2)T-n\right)v+T(n+2)-2n,\\
C_{p,d} &=&2\left(Tv^4+2(2T-1)v^3+2(T+1)v^2-2(2T-1)v+T-2\right) ,
\end{eqnarray*}%
and
\begin{equation*}
C_{s}=\frac{\alpha v^{2}}{bn}\frac{C_{p,n}}{\left( v-1\right) ^{2}\left(
v^{2}+2v-1\right) ^{2}}.
\end{equation*}

The corresponding quadratic differential form equals%
\begin{equation*}
\kappa =-\frac{Rn}{2}\frac{dT^{2}}{T^{2}}-\frac{R}{T}\frac{Tv^{4}+\left(
4T-2\right) v^{3}+\left( 2T+2\right) v^{2}+\left( 2-4T\right) v+T-2}{\left(
v-1\right) ^{2}\left( v^{2}+2v-1\right) ^{2}}dv^{2}.
\end{equation*}%
Therefore, the domain of applicable states is given by the inequality%
\begin{equation*}
Tv^{4}+\left( 4T-2\right) v^{3}+\left( 2T+2\right) v^{2}+\left( 2-4T\right)
v+T-2>0,
\end{equation*}%
and the phase transitions occur near the curve%
\begin{equation*}
T=\frac{2(v+1)(v-1)^{2}}{(v^{2}+2v-1)^{2}}.
\end{equation*}
\begin{figure}[h]
\centering
\includegraphics[scale=.4]{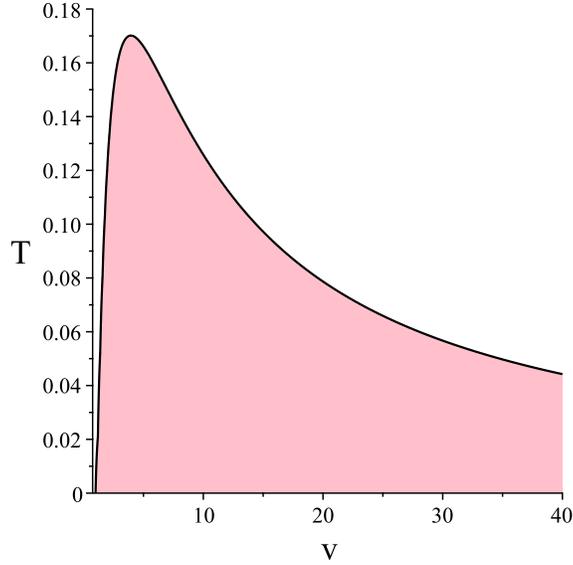}
\caption{Applicable domain for Peng-Robinson gases.}
\label{PRcurve}
\end{figure}
White domain in figure~\ref{PRcurve} corresponds to the domain of applicable states.
\section{Steady adiabatic filtration of real gases}

\subsection{Basic equations}

Steady gas filtration in homogeneous porous 3-dimensional media is described
by the following system of differential equations~\cite{Lib,Mus,Sch}:

\begin{itemize}
\item Darcy law or conservation of momentum%
\begin{equation}
\mathbf{U}=-\frac{k}{\mu }(\nabla _{p}),  \label{Darcy2}
\end{equation}%
for gas flow in anisotropic porous media, where $\rho $ is the density of
gas, $p$ -- pressure, $\mu $ -- dynamic viscosity, $k=\left\Vert
k_{ij}\right\Vert $ -- permeability tensor, $\mathbf{U}$ -- vector field of
filtration rate , and $\nabla _{p}$ -- the gradient of pressure.

\item conservation of mass%
\begin{equation}
\mathrm{div}\left( \rho \mathbf{U}\right) =0,  \label{mass}
\end{equation}%
where $\rho =v^{-1}.$

\item In addition to these equations, we assume that the filtration is adiabatic,
\begin{equation}
\mathbf{U}\left( \sigma \right) =0,  \label{adiabatic}
\end{equation}%
where $\sigma $ is a specific entropy, and by $\mathbf{U}\left( f\right) $
we denoted the derivative of a function $\ f$ along vector field $\mathbf{U.}
$

\item It is worth to note that the last condition in the cases of sources
leads us to local constancy of entropy in neighborhoods of the sources. We
extend this observation and propose some kind of \textquotedblleft
ergodicity hypothesis\textquotedblright . \ Namely, we'll consider a domain $%
\mathbf{D}\subset \mathbb{R}^{3}$ with sources $\left\{ a_{i}\right\} ,$
having the common constant specific entropy $\sigma _{0}.$

More precisely, consider a domain $\mathbf{D}\subset \mathbb{R}^{3}$ with sources $%
\left\{ a_{i}\right\} ,$ then under condition (\ref{adiabatic}) we represent
$\mathbf{D}$ as a union of domains $\mathbf{D}_{k}$, $k=1,2,..$, such that the
ergodicity hypothesis holds for each domain $D_{k},$ i.e. all sources in
this domain have the same entropy, and the rest$\ \mathbf{D}^{\prime }=%
\mathbf{D}\setminus \cup \mathbf{D}_{k}$ contains no sources. Filtrations in
these domains are independent. Thus, we may restrict ourselves by the case
of domains, where the ergodicity hypothesis holds.

\item The permeability tensor $k$ is a symmetric positive tensor depending
on media, $k=k\left( v,T\right) ,$ as well as viscosity $\mu .$
\end{itemize}

\subsection{Filtration Equations}

\begin{proposition}
The following formulae
\begin{equation*}
dp\equiv -v^{-2}c^{2}\ dv\mathrm{mod}\left\langle d\sigma \right\rangle ,
\end{equation*}%
and
\begin{equation*}
v^{-2}c^{2}=R\frac{T^{2}\left( \phi _{Tv}^{2}-\phi _{TT}\phi
_{vv}\right) +2T\left( \phi _{v}\phi _{vT}-\phi _{T}\phi _{vv}\right) +\phi
_{v}^{2}}{2\phi _{T}+T\phi _{TT}}
\end{equation*}%
are valid.
\end{proposition}

\begin{proof}
First of all, by the definition of the speed of sound we have%
\begin{equation*}
dp\equiv c^{2}d\left( v^{-1}\right) \mathrm{mod}\left\langle d\sigma
\right\rangle \equiv -v^{-2}c^{2}dv\mathrm{mod}\left\langle d\sigma
\right\rangle .
\end{equation*}

The second relation follows from the above formula for $C_{s}$
\end{proof}

Define a new tensor $\left\Vert Q_{ij}\left( v,\sigma _{0}\right)
\right\Vert $ as follows.

Assume that a fixed level $\sigma_{0}$ of entropy $\sigma $ is given and define function $%
T=\tau \left( v,\sigma _{0}\right) $ as a solution of the equation $\sigma
=R\sigma _{0},$ i.e.%
\begin{equation*}
\phi +T\phi _{T}=\sigma _{0}.
\end{equation*}%
The derivative of the left hand side of this equation in $T$ equals $%
R^{-1}\varepsilon _{T}$, that is positive in applicable domain, and
therefore the function $\tau $ exists and smooth.

Now we put%
\begin{equation}
Q_{ij}\left( v,\sigma _{0}\right) =\int \frac{c^{2}k_{ij}}{v^{3}\ \mu }dv,
\label{Qfunction}
\end{equation}%
where all functions under integral are functions in $v.$

As the integral of positive tensor the above tensor is also positive and
symmetric.

\begin{theorem}
Basic equations of adiabatic filtration with a given level of specific
entropy $R\sigma _{0}$ are equivalent to equation%
\begin{equation}
\sum_{i,j}\frac{d^{2}Q_{ij}}{dx_{i}dx_{j}}=0.  \label{QLaplace}
\end{equation}
\end{theorem}

\begin{proof}
Indeed,%
\begin{eqnarray*}
-\mathrm{div}\left( \rho \frac{k}{\mu }\left( \nabla _{p}\right) \right) &=&-%
\mathrm{div}\left( \sum_{ij}\frac{k_{ij}}{v\mu }\frac{\partial p}{\partial
x_{i}}\frac{\partial }{\partial x_{j}}\right) =\mathrm{div}\left( \sum_{ij}%
\frac{c^{2}k_{ij}}{v^{3}\mu }\frac{\partial v}{\partial x_{i}}\frac{\partial
}{\partial x_{j}}\right) = \\
\mathrm{div}\left( \sum_{ij}\frac{dQ_{ij}}{dx_{i}}\frac{\partial }{\partial
x_{j}}\right) &=&\sum_{i,j}\frac{d^{2}Q_{ij}}{dx_{i}dx_{j}}.
\end{eqnarray*}
\end{proof}

Let's rewrite equation (\ref{QLaplace}) in more details. To this end let us
introduce two symmetric matrices:%
\begin{equation*}
\mathrm{Hess}v=\left\Vert \frac{\partial ^{2}v}{\partial x_{i}\partial x_{j}}%
\right\Vert ,\ \left( \frac{\partial v}{\partial x}\right) ^{2}=\left\Vert
\frac{\partial v}{\partial x_{i}}\frac{\partial v}{\partial x_{j}}%
\right\Vert .
\end{equation*}

Then,%
\begin{equation*}
\frac{dQ_{ij}}{dx_{k}}=Q_{ij}^{\prime }\frac{\partial v}{\partial x_{k}},\quad%
\frac{d^{2}Q_{ij}}{dx_{k}dx_{l}}=Q_{ij}^{\prime }\frac{\partial ^{2}v}{%
\partial x_{k}\partial x_{l}}+Q_{ij}^{\prime \prime }\ \frac{\partial v}{%
\partial x_{k}}\frac{\partial v}{\partial x_{l}}\ ,
\end{equation*}%
and therefore equation (\ref{QLaplace}) takes the form%
\begin{equation}
\mathrm{Tr}\left( Q^{\prime }\cdot \mathrm{Hess}v+Q^{\prime \prime }\cdot \left(
\frac{\partial v}{\partial x}\right) ^{2}\right) =0,  \label{QLaplaceD}
\end{equation}%
where $Q^{\prime }=\left\Vert Q_{ij}^{\prime }\right\Vert $ and $Q^{\prime
\prime }=\left\Vert Q_{ij}^{\prime \prime }\right\Vert .$

The following two cases shall be important for us.

\begin{enumerate}
\item The porous media is isotropic, i.e. $k_{ij}=k\delta _{ij},$ where $%
\delta _{ij}$ is the Kronecker symbol. Then,
\begin{equation*}
Q_{ij}=Q\ \delta _{ij},
\end{equation*}%
where
\begin{equation*}
Q\left( v\right) =\int \frac{c^{2}k}{v^{3}\ \mu }dv.
\end{equation*}

\begin{theorem}
Basic equations of adiabatic filtration in isotropic media with a given
level of specific entropy $R\sigma _{0}$ are equivalent to equation%
\begin{equation}
\Delta (Q\left( v\right) )=0,  \label{QLaplaceIso}
\end{equation}%
where $\Delta $ is the Laplace operator.
\end{theorem}

\item The porous media is non isotropic but homogeneous, i.e. $k_{ij}$ are
constants.

Then%
\begin{equation*}
Q_{ij}=k_{ij}q,
\end{equation*}%
where%
\begin{equation*}
q=\int \frac{c^{2}}{v^{3}\ \mu }dv.
\end{equation*}

\begin{theorem}
Basic equations of adiabatic filtration in non isotropic but homogeneous
media with a given level of specific entropy $R\sigma _{0}$ are equivalent
to equation%
\begin{equation}
\mathrm{Tr}\left( k\cdot \mathrm{Hess}\left( q\right) \right) =0,
\label{QLaplaceHom}
\end{equation}%
where $\mathrm{Hess}\left( q\right) $ is the Hessian matrix of $q.$
\end{theorem}

\begin{remark}
By applying an orthogonal transformation we can transform matrix $k$ to the
diagonal form and in these coordinates equation (\ref{QLaplaceHom}) takes
the form%
\begin{equation}
k_{1}\frac{d^{2}q}{dx_{1}^{2}}+k_{2}\frac{d^{2}q}{dx_{2}^{2}}+k_{3}\frac{%
d^{2}q}{dx_{3}^{2}}=0,  \label{QLaplaceHomRed}
\end{equation}%
where $k_{1},k_{2},k_{3}$ are eigenvalues of $k.$
\end{remark}

\begin{remark}
A solution of all these equations gives us $v$ and $T$ as functions in the
domain $\mathbf{D}$ and by using the state equations we'll find also $p$ and $%
\varepsilon .$ Substituting these values in the equations of the phase
transitions curves give us equations for points in $\mathbf{D}$ where we
expect phase transitions.
\end{remark}
\end{enumerate}

\subsection{Integration of basic equations for isotropic and homogenous media%
}

Here we consider the isotropic case, the case of homogeneous media may be
elaborated in the similar manner.

Formula (\ref{QLaplaceIso}) shows that solutions of the basic system of
equations for adiabatic filtration (of the given level of entropy $\sigma
_{0}$) \ could be found in the following way.

Take a harmonic function $u$ in some domain $\mathbf{D}\subset %
\mathbb{R}^{3},$ and define, in general multivalued, function $v$ such that
\begin{equation*}
u=Q\left( v,\sigma _{0}\right) .
\end{equation*}%
Remark, that function $u$ is defined outside of a subset $\Sigma _{u}\subset
\mathbf{D}$ of singular points, where $\left\{ c^{2}=0\right\} ,$ and this
set consists of points in $D$, where $\left( v,T\right) \in \Sigma _{i}$.

Outside of this singular set all branches of the multivalued function $v$
are smooth and satisfy the basic equations. Also phase transitions occur at
points $x^{\left( 1\right) }$ and $x^{\left( 2\right) }$, \ where
corresponding values of one of branches of $v$ satisfy equations (\ref%
{PhaseEqui}).

\subsubsection{Model of source.}

Let's consider first solutions that correspond to sources of given intensity
$I.$ It means that
\begin{equation*}
u\left( x\right) =\frac{I}{4\pi }\left\vert x-a\right\vert ^{-1},
\end{equation*}%
if the source is located at point $a\in \mathbf{D,}$ or
\begin{equation}
v(x)=Q^{-1}\left( \frac{I}{4\pi }\left\vert
x-a\right\vert ^{-1},\sigma _{0}\right) .  \label{solutinf1}
\end{equation}%
This is a multivalued solution such that its branches are smooth outside of $%
\Sigma _{u}$ and phase transitions occur at points where equations (\ref%
{PhaseEqui}) hold.

\subsubsection{Dirichlet boundary problem.}

Consider an open and connected domain $\mathbf{D}\subset \mathbb{R}^{3}$
with a smooth boundary $\partial \mathbf{D,}$ equipped with a set $\mathbf{A}%
=\left\{ a_{i},i=1,\dots ,N\right\} \subset \mathbf{D}$ of points, which are the locations of sources, with
given intensities $I_{i}.$

We are looking for a solution $v$ of Dirichlet boundary problem of the basic
system in domain $\mathbf{D\smallsetminus A,}$ for a real gas with
Massieu-Plank potential function\textit{\ }$\phi \left( T,v\right) $ ,
having given intensities $I_{i}$ at points $a_{i},$ given entropy level $%
\sigma _{0},$ and given values of specific volume (or temperature) on the
boundary:
\begin{equation*}
\left. v\right\vert _{\partial \mathbf{D}}=v_{0}.
\end{equation*}%
To this end we take a harmonic in the domain $\mathbf{D\smallsetminus A}$
function $u$ such that
\begin{equation*}
u=\sum_{i}\frac{I_{i}}{4\pi }\left\vert x-a_{i}\right\vert^{-1} +u_{0},
\end{equation*}%
where $u_{0}$ is the harmonic function in $\mathbf{D}$ with the following
boundary conditions%
\begin{equation*}
\left. u_{0}\right\vert _{\partial \mathbf{D}}=Q\left( v_{0}\right) -\left.
\left( \sum_{i}\frac{I_{i}}{4\pi }\left\vert x-a_{i}\right\vert^{-1} \right)
\right\vert _{\partial \mathbf{D}}.
\end{equation*}%
Then the multivalued function $Q^{-1}\left( u\right) $ gives us the solution
of the Dirichlet boundary problem.
\subsection{Examples}
Here we apply above methods to the van der Waals and Peng-Robinson gases filtration.
\subsubsection{van der Waals gases.}
Given level of the specific entropy $\sigma_{0}$ allows us to express the temperature $T$ and the pressure $p$ as functions of the specific volume $v$ using~(\ref{phiandent}). For van der Waals gases we get:
\begin{equation*}
T(v)=c(3v-1)^{-\alpha},\qquad p(v)=8c(3v-1)^{-\alpha}-\frac{3}{v^{2}},
\end{equation*}
where
\begin{equation*}
\alpha=1+\frac{2}{n},\qquad c=\exp\left( \frac{3\sigma_{0}}{4n} \right).
\end{equation*}
Function $Q(v)$ is defined by the following relation:
\begin{equation*}
-\frac{\mu}{k}Q(v)=-\frac{2}{v^{3}}+8c\frac{(3v-1)^{-\alpha}}{v}+8c\int(3v-1)^{-\alpha}v^{-2}dv.
\end{equation*}
As it has already been shown, to get an explicit solution we have to invert $Q(v)$. For this reason it is essential for us to figure out the invertibility conditions for $Q(v)$.
\begin{theorem}
Function $Q(v)$ is invertible if the specific entropy constant $c$ satisfies the following inequality:
\begin{equation*}
c>\frac{1}{4\alpha}(1+\alpha)^{1+\alpha}(2-\alpha)^{2-\alpha}.
\end{equation*}
\end{theorem}
\begin{proof}
We need to define a condition when the function $Q(v)$ is monotonic if $v>1/3$. This means that its derivative must not be equal to zero. As the conditions $Q^{\prime}(v)=0$ and $p^{\prime}(v)$ are equivalent, it easy to check that
\begin{equation*}
Q^{\prime}(v)=0\Longleftrightarrow w(v)=4c\alpha,
\end{equation*}
where $w(v)=(3v-1)^{1+\alpha}v^{-3}$. Function $w(v)$ has a maximum at the point $v_{0}$, which is derived from the equation $w^{\prime}(v)=0$ and it equals
\begin{equation*}
v_{0}=\frac{1}{2-\alpha}.
\end{equation*}
Therefore, $w(v_{0})=(1+\alpha)^{1+\alpha}(2-\alpha)^{2-\alpha}$.
\end{proof}

Having a solution we can move the coexistence curve from the space of thermodynamical variables to the $\mathbb{R}^{3}(x_{1},x_{2},x_{3})$. Let's suppose that we have a point source at the point $x_{0}$. Then, in case of invertible function $Q(v)$ we have the picture presented by figure~\ref{phspace1}.

\begin{figure}
\centering
\includegraphics[scale=.3]{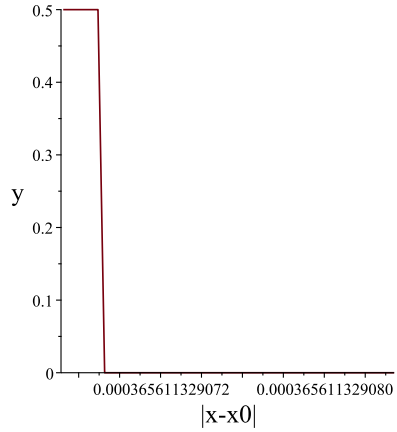}\\
\caption{Phases in space for van der Waals gases if $Q(v)$ is invertible. If variable $y=0$, the domain $|x-x_{0}|$ corresponds to the gas phase. If $y=0.5$, the domain $|x-x_{0}|$ corresponds to the intermediate state (condensation process).}
\label{phspace1}
\end{figure}

Figure~\ref{phspace1} shows that in the neighbourhood of the source condensation of the gas is observed.

In case of irreversibility of the function $Q(v)$ the picture is much more complicated. First of all, solution is multivalued. It is shown in figure~\ref{tempmulti}. This means that there is a number of possibilities for the gas filtration development.

\begin{figure}[ht!]
\centering \subfigure[]{
\includegraphics[width=0.4\linewidth]{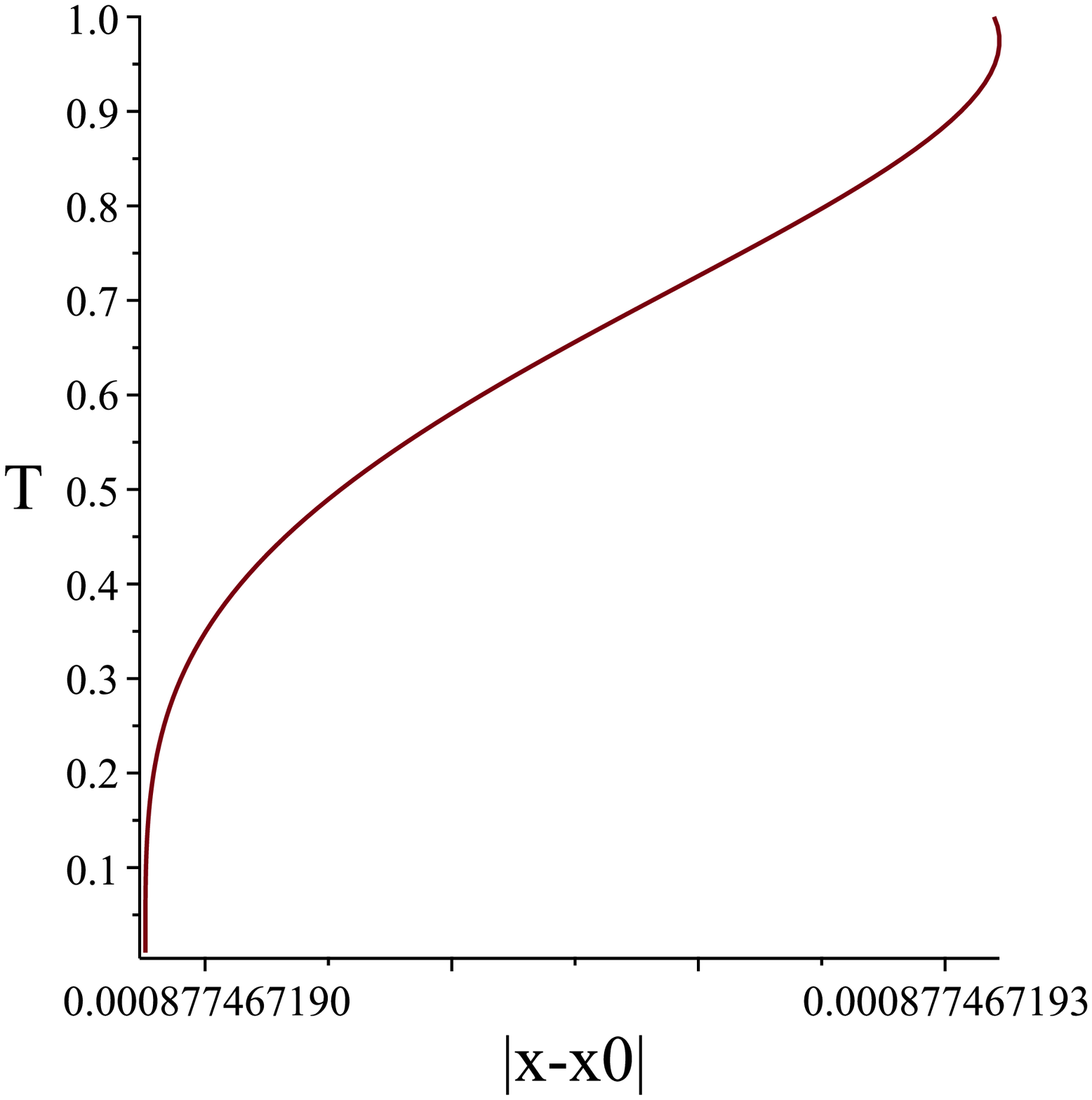} \label{templow} }
\hspace{4ex}
\subfigure[]{ \includegraphics[width=0.4\linewidth]{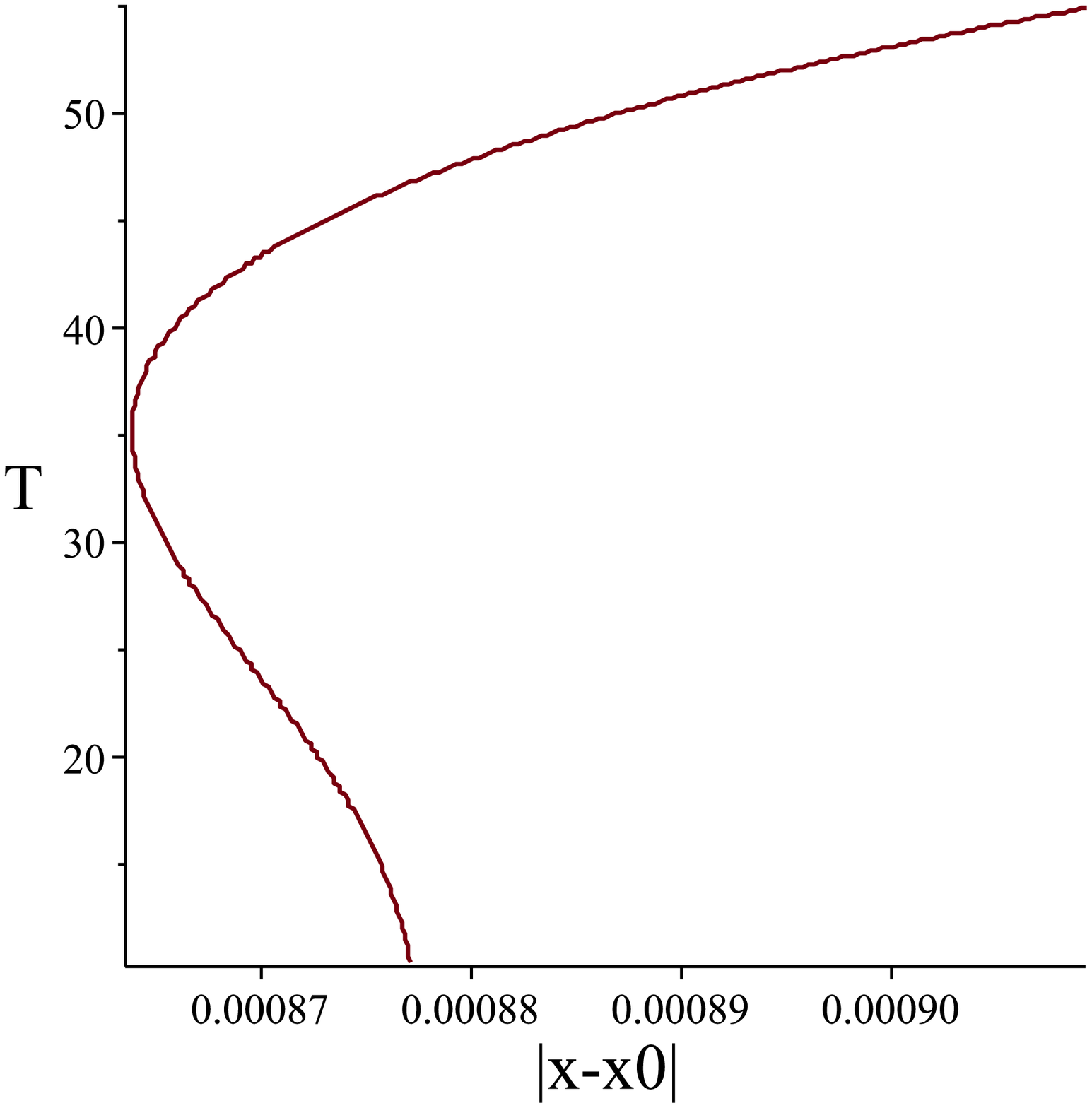} \label{temphigh} }
\caption{Dependence of temperature on the distance from the source for van der Waals gases: \subref{templow} represents low temperatures; \subref{temphigh} represents high temperatures.} \label{tempmulti}
\end{figure}

In figure~\ref{tempmulti} the temperature under the critical point is uniquely determined, which leads to the uniqueness in phase under the critical point. It is shown in  figure~\ref{phspace2}. We can see that the picture is similar to the previous one.

\begin{figure}
\centering
\includegraphics[scale=.35]{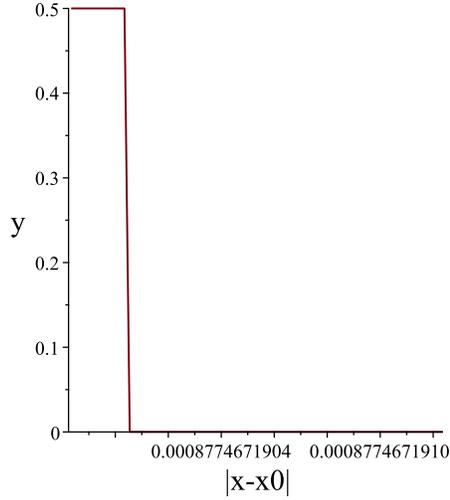}\\
\caption{Phases in space for van der Waals gases if $Q(v)$ is irreversible. If variable $y=0$, the domain $|x-x_{0}|$ corresponds to the gas phase. If $y=0.5$, the domain $|x-x_{0}|$ corresponds to the intermediate state (condensation process).}
\label{phspace2}
\end{figure}

But in figure~\ref{tempmulti1} the temperature under the critical point is multivalued. This means that we can expect that at the same point there can be different phases. It is shown in figure~\ref{phspace3}.

\begin{figure}[ht!]
\centering \subfigure[]{
\includegraphics[width=0.4\linewidth]{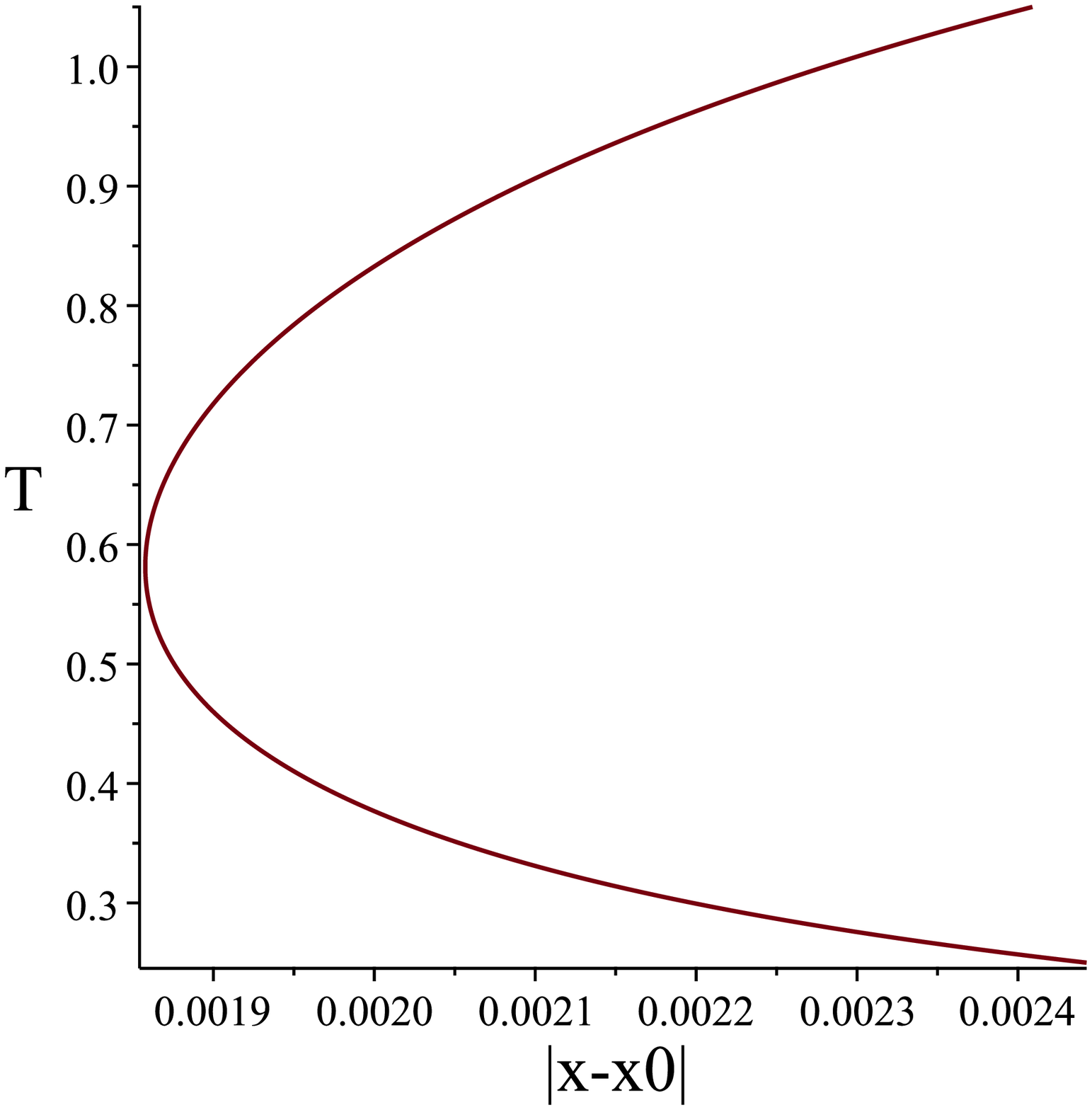} \label{templow1} }
\hspace{4ex}
\subfigure[]{ \includegraphics[width=0.4\linewidth]{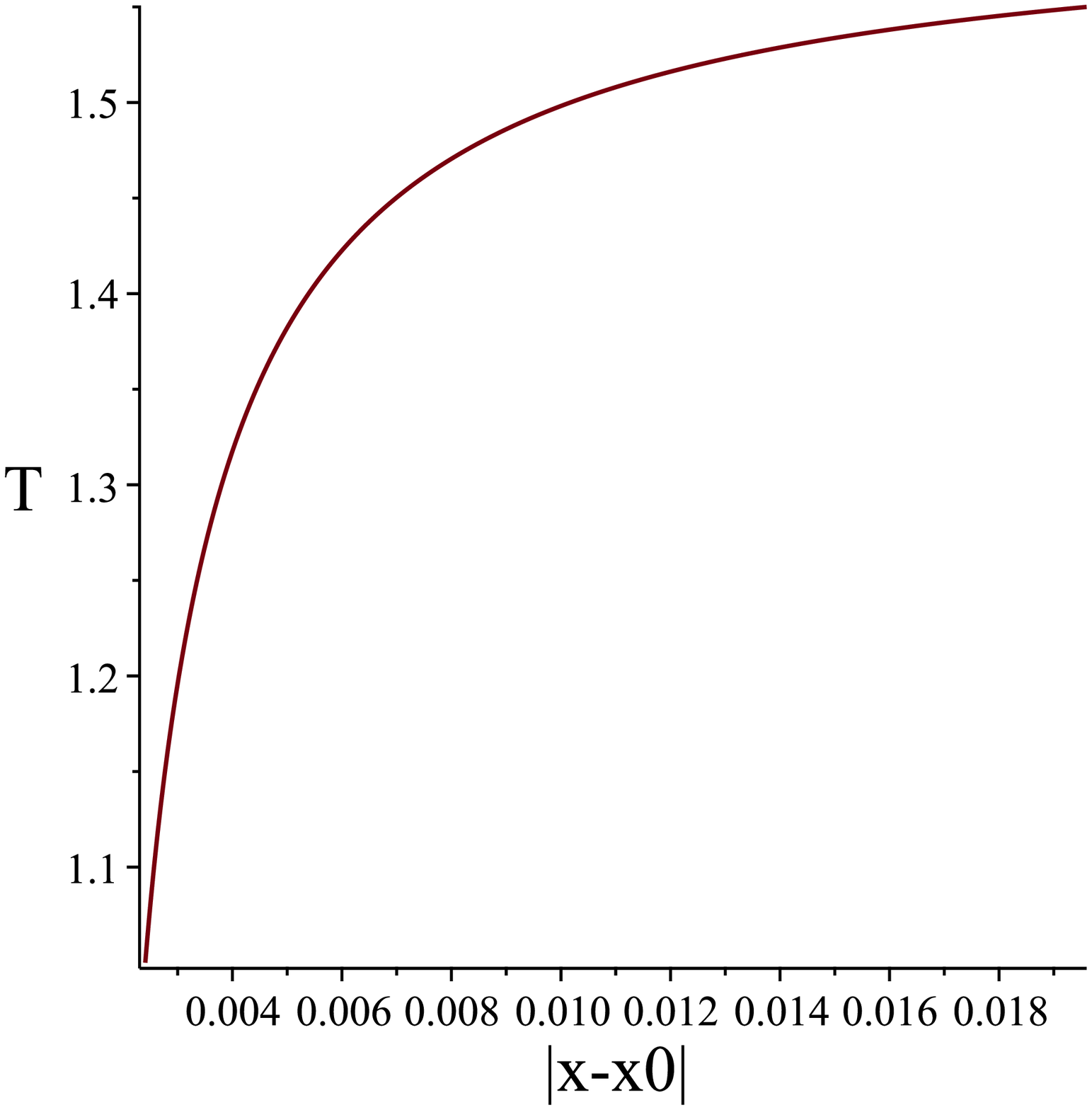} \label{temphigh1} }
\caption{Dependence of temperature on the distance from the source for van der Waals gases: \subref{templow1} represents low temperatures; \subref{temphigh1} represents high temperatures.} \label{tempmulti1}
\end{figure}

\begin{figure}
\centering
\includegraphics[scale=.3]{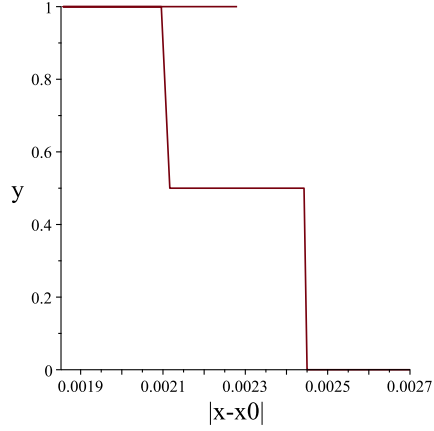}\\
\caption{Phases in space for van der Waals gases if $Q(v)$ is irreversible. If variable $y=0$, the domain $|x-x_{0}|$ corresponds to the gas phase, if $y=0.5$, the domain $|x-x_{0}|$ corresponds to the intermediate state (condensation process), if $y=1$, the domain $|x-x_{0}|$ corresponds to the liquid phase.}
\label{phspace3}
\end{figure}

\subsubsection{Peng-Robinson gases.}
In case of Peng-Robinson gases the pressure $p$ and the temperature $T$ can be expressed as functions of the specific volume $v$ as follows:
\begin{equation*}
T(v)=c(v-1)^{-\alpha+1},\qquad p(v)=c(v-1)^{-\alpha}-\frac{1}{v^{2}+2v-1},
\end{equation*}
where
\begin{equation*}
c=\exp\left( \frac{2\sigma_{0}}{n} \right),\qquad\alpha=1+\frac{2}{n}.
\end{equation*}
Function $Q(v)$ is defined by the following relation:
\begin{equation*}
\begin{split}
-\frac{\mu}{k}Q(v)=&\frac{c}{v(v-1)^{\alpha}}-\frac{v+2}{v^{2}+2v-1}+\ln\left(\frac{v^{2}}{v^{2}+2v-1}\right)+{}\\&+\frac{3}{\sqrt{2}}\mathrm{arctanh}\left(\frac{v+1}{\sqrt{2}}\right)+c\int(v-1)^{-\alpha}v^{-2}dv.
\end{split}
\end{equation*}
Finally, invertibility condition can be formulated as follows.
\begin{theorem}
Function $Q(v)$ is invertible if the specific entropy constant $c$ satisfies the following inequality:
\begin{equation*}
c>\frac{2}{\alpha}\frac{(v_{0}+1)(v_{0}-1)^{\alpha+1}}{(v_{0}^{2}+2v_{0}-1)^{2}},
\end{equation*}
where $v_{0}$ is the root of the equation:
\begin{equation*}
(\alpha-2)v^{3}+3\alpha v^{2}+(\alpha+2)v-\alpha+4=0.
\end{equation*}
There exists a real root of the above equation if $\alpha<2$, or, equivalently, $n>2$.
\end{theorem}

\section*{Acknowledgments}
This work was supported by the Russian Foundation for Basic Research (project No 18-29-10013).


\begin{thebibliography}{99}
\bibitem{GLRT} Gorinov A, Lychagin V, Roop M and Tychkov S 2019 Gas flow with phase transitions: thermodynamics and the
Navier-Stokes equations \textit{Nonlinear PDEs, their geometry and applications. Proceedings of Wisla 18 Summer School} (Switzerland: Springer Nature) 229--241

\bibitem{LychGSA} Lychagin V 2019 Adiabatic filtration of ideal gases in
a homogeneous and isotropic porous media \textit{Global and Stochastic Analysis}
vol 6(1)

\bibitem{Lib} Leibenson L S 1947 Motion of natural liquids and gases in a porous
medium \textit{Gostkhizdat} Moscow

\bibitem{Mus} Muskat M 1937 The Flow of Homogeneous Fluids Through Porous
Media (New York: McGraw-Hill)

\bibitem{LY} Lychagin V 2019 Contact Geometry, Measurement and
Thermodynamics. \textit{Nonlinear PDEs, their geometry and applications. Proceedings of Wisla 18 Summer School} (Switzerland: Springer Nature) 3--54

\bibitem{KLR} Kushner A, Lychagin V and Roubtsov V 2007 Contact geometry and nonlinear differential equations (Cambridge: Cambridge University Press)

\bibitem{Fort} Fortov V 2011 Equation of state \textit{Thermopedia} (doi:
10.1615/AtoZ.e.equation\_of\_state)

\bibitem{Tol} Tolman Richard C 1917 The Measurable Quantities of Physics \textit{Phys.
Rev.} 9(3) 237--253

\bibitem{Onnes} Kamerlingh Onnes H 1902 Expression of state of gases and
liquids by means of series \textit{KNAW Proceedings} 4 Amsterdam 125--147 


\bibitem{PR} Peng D Y and Robinson D B 1976 A New Two-Constant Equation of
State \textit{Industrial and Engineering Chemistry: Fundamentals.} 15(1) 59--64 (doi:10.1021/i160057a011)

\bibitem{Sch} Scheidegger Adrian E. 1960 The physics of flow through porous
media. Revised edition. (New York: The Macmillan Co.) 313 pp

\bibitem{LychSing} Lychagin V 1985 Singularities of multivalued solutions
of nonlinear differential equations, and nonlinear phenomena \textit{Acta Appl.
Math.} 3(2) 135--173

\bibitem{ST} Stanley Harry Eugene 1971 Introduction to Phase Transitions and
Critical Phenomena (Oxford: Oxford University Press) 333 pp

\bibitem{DLT} Duyunova A, Lychagin V and Tychkov S 2017
Classification of equations of state for viscous fluids \textit{Doklady
Mathematics} 95(2) 172--175 (doi:10.1134/S1064562417020211)

\end{thebibliography}
\end{document}